%% file: ms.tex
\documentclass[journal]{IEEEtran}

\makeatletter

\newcommand{\Rmnum}[1]{\expandafter\@slowromancap\romannumeral #1@}
\makeatother

\usepackage{amsmath,bm,color}
\usepackage{amssymb}
\usepackage{algorithm}
\usepackage{algorithmic}
\usepackage{graphicx} 
\usepackage{epstopdf}
\usepackage{amsthm}
\usepackage{yhmath}
\usepackage{setspace}
\usepackage{color}
\usepackage{threeparttable}
\usepackage{lipsum}

\newcommand\blfootnote[1]{%
  \begingroup
  \renewcommand\thefootnote{}\footnote{#1}%
  \addtocounter{footnote}{-1}%
  \endgroup
}

\include{com}

\newcommand{\g}[1]{\bm #1}

\begin{document}

\title{Hybrid Interference Mitigation Using Analog Prewhitening }

\author{Wei Zhang \quad Yi Jiang \quad Bin Zhou \quad Die Hu 
}


\maketitle
\blfootnote{Work in this paper was supported by National Natural Science Foundation of China Grant No. 61771005. The partial material in the paper has be published in the 2020 IEEE Sensor Array and Multichannel (SAM) Signal Processing Workshop.

W. Zhang, Y. Jiang (the corresponding author), and D. Hu are with 
Key Laboratory for Information Science of Electromagnetic Waves (MoE), Department of Communication Science and Engineering, Fudan University, Shanghai, China (emails: 19110720048@fudan.edu.cn, yijiang@fudan.edu.cn, hudie@fudan.edu.cn).

B. Zhou is with the Key Laboratory of Wireless Sensor Network and Communications, Chinese Academy of Sciences (CAS) (e-mail: bin.zhou@mail.sim.ac.cn).
}
\begin{abstract}
This paper proposes a novel scheme for mitigating strong interferences, which is applicable to various wireless scenarios, including full-duplex wireless communications and uncoordinated heterogenous networks. As strong interferences can saturate the receiver's analog-to-digital converters (ADC), they need to be mitigated both before and after the ADCs, i.e., via hybrid processing. The key idea of the proposed scheme, namely the Hybrid Interference Mitigation using Analog Prewhitening (HIMAP), is to insert an $M$-input $M$-output analog phase shifter network (PSN) between the receive antennas and  the ADCs to spatially prewhiten the interferences, which requires no signal information but only an estimate of the covariance matrix. After interference mitigation by the PSN prewhitener, the preamble can be synchronized, the signal channel response can be estimated, and thus a minimum mean squared error (MMSE) beamformer can be applied in the digital domain to further mitigate the residual interferences. The simulation results verify that the HIMAP scheme can suppress interferences 80dB stronger than the signal by using off-the-shelf phase shifters (PS) of 6-bit resolution.
\end{abstract}


%
\IEEEpeerreviewmaketitle

\section{Introduction}
Multi-antenna beamforming for interference mitigation has been widely used in wireless communications over the ever-increasingly crowded frequency spectrum \cite{Tse2004Fundamentals}. The conventional interference mitigation methods, including the zero-forcing (ZF) and minimum mean squared error (MMSE) beamforming, are usually conducted in the digital domain.

In recent years, owing to the advent of full-duplex wireless communications  \cite{7105651}\cite{Sabharwal2014In} and heterogeneous network \cite{6211486}\cite{6231163}\cite{7060498}, mitigating exceedingly strong interferences has drawn considerable attentions.  A strong interference, e.g., over 70dB stronger than the signal, can saturate the analog-to-digital converters (ADC), rendering large quantization noise that can overwhelm the digital signal processing (DSP) only approaches. As to improve the ADCs' bit resolution can be too costly \cite{7258468}, such strong interferences have to be mitigated in both analog and digital domain.

Multiple approaches have been proposed to mitigate strong interferences in the full-duplex wireless researches \cite{101145}\cite{1232011456}\cite{6648617}\cite{5985554}. The paper \cite{101145} proposes an antenna cancellation technique, which uses two transmit antennas and one receive antenna to mitigate the self-interference by having the signals of two transmit antennas combine destructively at the receive antenna. 
 As the self-interference is known, it can also be mitigated by feeding from the digital baseband to the analog circuit, so that it can be largely cancelled in the time domain \cite{1232011456}\cite{6648617}.
In the full-duplex relay scenario, the paper \cite{5985554} proposes to use multiple antennas for spatial interference suppression, in addition to natural separation and  time domain cancellation. 


Interference mitigation is also a major challenge in heterogeneous network \cite{6211486}\cite{6231163}. The so-called inter-cell interference coordination (ICIC)  technique, which includes power control, fractional frequency reuse, interference regeneration and cancellation, is not always feasible, as it could incur excessive overhead of coordination \cite{6211486}. Indeed, a transmitting macrocell user equipment (UE) at the edge of the macro-cell may impose severe interferences to a nearby non-accessible microcell, causing the so-called ``deadzone'' \cite{6231163}. The interference mitigation techniques proposed for the full duplex wireless are not applicable here as the external interference -- either its time-domain waveform or spatial direction -- is unknown.

In this paper, we propose a new scheme named Hybrid Interference Mitigation using Analog Prewhitening (HIMAP). Assuming an $M$-antenna receiver, we insert an $M$-input $M$-output  phase shifter network (PSN)  between the antenna ports and the ADCs. The HIMAP optimizes the PSN to spatially prewhiten the received interferences in the radio frequency (RF) domain. We show that the PSN can suppress the interferences significantly before entering the ADCs, leading to much reduced quantization noise. Consequently, a standard multi-antenna technique, such as the MMSE beamforming,  can be employed to suppress the residual interferences in the digital domain.

Different from our HIMAP scheme that adopts a ``square'' $M\times M$ PSN,  the paper \cite{VenkateswaranVeen2010} proposes to use a ``rectangular'' $M\times N$ PSN ($N<M$) as a subspace steerer so that the interferences orthogonal to the PSN subspace will be eliminated. The similar idea of subspace projection and ``MMSE filtering'' is also proposed in \cite{5985554}. But such approaches need to know the signal channel response, which may not be feasible in the presence of strong interferences. The HIMAP scheme, however, requires only the covariance matrix estimate.

The rationale of applying the analog prewhitener is as follows. In the presence of strong interferences, the inputs into the multiple ADCs are spatially correlated, which suggests a waste of the bit resolution of the ADCs. With the analog PSN prewhitening, the inputs to the ADCs become uncorrelated, leading to efficient usage of the ADC bits.

It is worthwhile mentioning that PSNs are standard components in modern communication systems. Besides interference mitigation, the PSN-based hybrid beamforming has been proposed for massive multi-input multi-output (MIMO) \cite{LarssonEdforsTufvessonMarzetta2014}  for high spectral efficiency with reduced hardware cost \cite{7397861}\cite{8927851}
\cite{7389996}, where a ``rectangular'' PSN connects a large number of antenna ports to a smaller number of RF chains.

The simulation results verify that the HIMAP scheme can suppress interferences  80dB stronger than the signal-of-interest even using an off-the-shelf phase shifters (PS) of 6-bit resolution. Since the HIMAP scheme assumes no knowledge of the interference waveform, it works for not only full-duplex wireless communications but the non-cooperative scenarios, such as heterogenous networks and anti-jamming tactical communications.

The remainder of this paper is organized as follows. Section \ref{SEC2} establishes the system model and illustrates the conception of ideal prewhitening. Section \ref{SEC3} introduces the HIMAP as a five-step scheme and proposes a heuristic PSN-based method to accomplish prewhitening in analog domain. Numerical results are presented in Section \ref{SEC4} to demonstrate the effectiveness of the proposed scheme. The conclusions are drawn in Section \ref{SEC5}.

\textit{Notations:} $(\cdot)^H$ denotes the conjugate transpose, $ (\cdot)^*$ stands for complex conjugate. and $||\cdot||^2$, $\Enum(\cdot)$, $\tr(\cdot)$ represent a matrix's norm, determinant, and trace, respectively. $\text{diag}(\vbf)$ stands for a diagonal matrix with vector $\vbf$ being the diagonal element. $\ebf_n\in {\mathbb C}^N$ denotes an $N$-element vector with the $n$th element being one and others being zero. $|{\cal S}|$ is the cardinality of the set ${\cal S}$. $x\sim N(\mu,\sigma^2)$ is a Gaussian random variable with mean $\mu$ and variance $\sigma^2$.

\section{Signal Model and Preliminaries} \label{SEC2}
Consider the received signal vectors by an $M$-element antenna array
\ben \label{eq.y}
\ybf(t) = \hbf x(t) + \sum_{k=1}^K \gbf_k \xi_k(t) + \g{\eta}(t),
\een
where $\hbf \in {\mathbb C}^{M\times1}$ is the array response of the signal $x(t)$ with power ${\mathbb E}[|x(t)|^2]=\sigma_x^2$, $\xi_k(t)$ represents the $k$th interference,  $\gbf_k\in\Cnum^{M\times1}$ is its array response, and $\g{\eta}(t)\in\Cnum^{M\times1}$ is the thermal noise.

Denote $\zbf(t) \triangleq \sum_{k=1}^K \gbf_k \xi_k(t) + \g{\eta}(t)$ as the lump sum of the interference-plus-noise, whose covariance is $\Cbf_z = {\mathbb E}[\zbf(t)\zbf(t)^H] \in {\mathbb C}^{M\times M}$.

\subsection{Linear MMSE Beamforming}
To suppress the interferences, we adopt the linear MMSE beamformer \cite{Tse2004Fundamentals}
\ben \wbf = \sigma_x^2 \Rbf_y^{-1}\hbf, \label{eqwmmse}\een
where
\ben \Rbf_y = {\mathbb  E}[\ybf(t)\ybf(t)^H] = \sigma_x^2\hbf \hbf^H + \Cbf_z\een
 is the covariance matrix.

 Using the matrix inversion lemma \cite{Hager1989}, we have
 \ben
\Rbf_y^{-1} = \Cbf_z^{-1} - \Cbf_z^{-1}\hbf \hbf^H \Cbf_z^{-1}/(\sigma_x^{-2} + \hbf^H \Cbf_z^{-1} \hbf). \label{eqRy-1}
 \een
Combining (\ref{eqwmmse}) and (\ref{eqRy-1}) yields
 \ben \wbf = \alpha \Cbf_z^{-1}\hbf, \label{eqwmmseQ}\een
for some scalar $\alpha >0$.
 Applying $\wbf$ to the received vector $\ybf(t)$, we obtain
\ben
\hat{x}(t) = \wbf^H \ybf(t),
\een
for which the post-processing signal to interference-plus-noise ratio (PPSINR) is \cite{4799043}
\ben \label{sinrOUT}
\rho_{\rm mmse}  =
\sigma_x^2\hbf^H \Cbf_z^{-1} \hbf.
\een
The above expression assumes no estimation errors or ADC quantization noise.

In practice, the covariance matrix $\Rbf_y$ need to be estimated from the digital samples
\ben \ybf_n = {\cal Q}[\ybf(nT_s)] \label{eqyn}\een
and
\ben  \label{eq.Ry}
\hat{\Rbf}_y = \frac{1}{L}\sum_{n=1}^L \ybf_n \ybf_n^H,
\een
where $T_s$ is the sampling time interval and ${\cal Q}[\cdot]$ is the ADC quantization.
The array response $\hbf$ can be estimated based on the synchronized training sequence $\{ x_n, n=1,\ldots, L \}$ \cite{1254022}, i.e.,
\ben
\hat{\hbf} = \frac{1}{\sum_{n=1}^L |x_n|^2}\sum_{n=1}^L \ybf_n x_n^*.
\een
How to achieve the synchronization of the training sequence will be addressed in Section \ref{sec3.2}. The signal power $\sigma_x^2 = \frac{1}{L} \sum_{n=1}^{L} |x_n|^2 $.

\subsection{ADC Quantization Noise Due to Strong Interferences}
The presence of strong interferences, however, can induce excessive ADC quantization noise, which may void digital-only interference mitigation.
Indeed, the relation between ${\rm SQNR}_{\rm dB}$ and ${\rm SIR}_{\rm dB}$ is
\begin{equation}
{\sf SQNR}_{\rm dB} \approx {\sf SIR}_{\rm dB} + 6.02{\sf ENOB}-4.35,
\label{eq.SQNRdb}
\end{equation}
which is derived as follows.

The covariance matrix of quantization noise is \cite{7307134}
\begin{equation}
\Rbf_{\qbf}=\rho(1-\rho){\rm diag}(\Rbf_y),
\end{equation}
where $\rho = \frac{\pi\sqrt{3}}{2}2^{-2{\sf ENOB}}$, and ${\sf ENOB}$ stands for the effective number of bits of the ADC \cite{7258468}.
%
As the SINR at the antennas is
\begin{equation}
{\sf SINR} = \frac{\hbf^H\hbf\sigma_x^2}{\sum_{k=1}^K\gbf_k^H\gbf_k\sigma_{\xi_k}^2+M\sigma_{\eta}^2},
\label{equ.sinr}
\end{equation}
where $\sigma_{\xi_k}^2$ and $\sigma_{\eta}^2$ represent the power of the $k$th interferences and the thermal noise, respectively, the signal to quantization noise ratio (SQNR) can be expressed as
\begin{equation}
\begin{split}
{\sf SQNR} &= \frac{{\rm tr}(\hbf\hbf^H\sigma_x^2)}{\rho(1-\rho){\rm tr}(\Rbf_{y})}
\\&= \frac{1}{\rho(1-\rho)}\times\frac{\hbf^H\hbf\sigma_x^2}{\hbf^H\hbf\sigma_x^2+\sum_{k=1}^K\gbf_k^H\gbf_k\sigma_{\xi_k}^2+M\sigma_{\eta}^2}
\\& = \frac{1}{\rho(1-\rho)}\times\frac{1}{1+\frac{1}{\sf SINR}}
\\& \approx \frac{1}{\rho(1-\rho)}\times{\sf SINR},
\end{split}
\label{equ.Sqr}
\end{equation}
where the last approximation of (\ref{equ.Sqr}) holds for ${\sf SINR \ll 1}$ in the presence of strong interference. Thus, we can approximate ({\ref{equ.Sqr}}) as
\begin{equation}
{\sf SQNR}_{\rm dB} \approx {\sf SINR}_{\rm dB} - 10{\rm lg}(\rho - \rho^2).
\label{equ.sqrdB}
\end{equation}
Substituting $\rho = \frac{\pi\sqrt{3}}{2}2^{-2{\sf ENOB}}$ and further ignoring the higher-order term $\rho^2$, we can further express (\ref{equ.sqrdB}) as
\begin{equation}
{\sf SQNR}_{\rm dB} \approx {\sf SINR}_{\rm dB} + 6.02{\sf ENOB}-4.35.
\end{equation}
As ${\sf SIR}_{\rm dB} \approx {\sf SINR}_{\rm dB}$ for strong interferences, we have obtained (\ref{eq.SQNRdb}).

Given a strong  interference, e.g., ${\sf SIR}_{\rm dB} = -80$,  a ADC with ENOB $12$bit will yield SQNR only $-12.1$dB. Thus, the PPSINR formula (\ref{sinrOUT}) would be too optimistic given the finite-bit ADCs..

To have a larger ENOB can help, but with very steep increase of the ADC power consumption -- 4x increase per extra bit \cite{7258468}. Thus, the strong interferences have to be mitigated in both analog and digital domain, i.e., via hybrid interference mitigation.

\subsection{Interference Mitigation via Spatial Prewhitening}
We propose to split the MMSE beamformer (\ref{eqwmmse}) into two parts:
\ben
\wbf = \left(\Rbf_y^{-1/2}\right)\left(\sigma_x^2 \Rbf_y^{-1/2} \hbf\right).
\een
The first term $\Rbf_y^{-1/2}$ corresponds to the analog spatial prewhitening before the ADCs, while the second term $(\sigma_x^2 \Rbf_y^{-1/2} \hbf)$ amounts to the digital beamforming after the ADCs.
$\Rbf_y^{-1/2}$ should satisfy
\ben
\Rbf_y^{-1/2}\Rbf_y(\Rbf_y^{-1/2})^{H} = \Ibf.
\een
Thus $\Rbf_y^{-1/2}$ can be computed as
\ben
\Rbf_y^{-1/2} = \Qbf\Sigmabf^{-1/2}\Ubf^{H}
\label{equ.RnegHalf}
\een
where $\Ubf$ and $\g{\Sigma}$ are from the SVD $\Rbf_y=\Ubf\Sigmabf\Ubf^{H}$, and $\Qbf$ can be an arbitrary unitary matrix; thus, $\Rbf_y^{-1/2}$ is not unique.

\begin{Lemma} \label{lemma1}
 The signal to interference-plus-noise ratio (SINR) of the input to the ADCs after the prewhitener $\Rbf_y^{-1/2}$ is
\ben \label{eqrhoprewhite}
\rho_{\rm prewhite}  = \frac{\frac{\rho_{\rm mmse}}{1+\rho_{\rm mmse}}}{M-\frac{\rho_{\rm mmse}}{1+\rho_{\rm mmse}}},
\een
where $\rho_{\rm mmse}=
\sigma_x^2\hbf^H \Cbf_z^{-1} \hbf$ is as given in (\ref{sinrOUT}).
\end{Lemma}
\proof The prewhitening renders the signal array response to be $\Rbf_y^{-1/2} \hbf$, and the lump-sum of the interference-plus-noise to be $\Rbf_y^{-1/2} \zbf$. By the definition of SINR
\begin{align}
\rho_{\rm prewhite} &= \frac{\|\Rbf_y^{-1/2}\hbf \|^2 \sigma_x^2}{ \tr (\Rbf_y^{-1/2} \Cbf_z \Rbf_y^{-1/2})}= \frac{\hbf^H \Rbf_y^{-1}\hbf \sigma_x^2}{ \tr (\Rbf_y^{-1} \Cbf_z)} \label{eq.pre-ADC}\\
 &= \frac{\hbf^H \Rbf_y^{-1}\hbf \sigma_x^2}{ \tr (\Rbf_y^{-1}(\Rbf_y -  \hbf\hbf^H \sigma_x^2))} \nonumber \\
&=\frac{\hbf^H \Rbf_y^{-1}\hbf \sigma_x^2}{M-\hbf^H \Rbf_y^{-1}\hbf \sigma_x^2}. \label{eqrhoprewhite2}
\end{align}
Using (\ref{eqRy-1}), we have
\begin{align}
\hbf^H \Rbf_y^{-1}\hbf\sigma_x^2 &= \sigma_x^2\hbf^H\Cbf_z^{-1}\hbf - \frac{\sigma_x^2|\hbf^H\Cbf_z^{-1}\hbf |^2}{\sigma_x^{-2} + \hbf^H \Cbf_z^{-1} \hbf} \nonumber\\
&=\frac{\sigma_x^2\hbf^H\Cbf_z^{-1}\hbf }{1+\sigma_x^{2}\hbf^H \Cbf_z^{-1} \hbf} = \frac{\rho_{\rm mmse}}{1+\rho_{\rm mmse}}. \label{eqrhoprewhite3}
\end{align}
Combining (\ref{eqrhoprewhite3}) and (\ref{eqrhoprewhite2}) yields (\ref{eqrhoprewhite}).
\endproof
From Lemma \ref{lemma1} we see that $\rho_{\rm prewhite} \approx \frac{1}{M-1}$ if $\rho_{\rm mmse} \gg 1$. Hence, the prewhitening amounts to mitigating the strong interference in the spatial domain, which leads to much reduced ADC quantization noise according to (\ref{eq.SQNRdb}).


%
%

Another way to see why prewhitening can reduce the ADC quantization noise is as follows. In the presence of strong interferences, the received waveforms at different antennas are highly correlated, leading to a waste of ADC bit resolutions. The prewhitening decorrelates the waveforms received by the different antennas, leading to more efficient usage of the ADC bits.

A major advantage of prewhitening is that it only requires the covariance matrix estimate $\hat{\Rbf}_y$, while the MMSE method in \cite{VenkateswaranVeen2010} also needs to know $\hbf$, which is difficult in the presence of strong interferences.

\subsection{Phase Shifter Network}
It is difficult to realize the ideal prewhitener $\Rbf_y^{-1/2}$ using analog circuits. We propose to employ an $M\times M$ PSN in between the antenna ports and the ADCs to approximate the ideal prewhitener. The PSN can be represented by an $M\times M$ matrix,
\ben \label{EphiStruc}
\Ebf(\Phi) = \left(\begin{matrix}
1 & e^{j\phi_{1,2}}& \ldots & e^{j\phi_{1,M}} \\
e^{j\phi_{2,1}} & 1 & \ldots & e^{j\phi_{2,M}}\\
\vdots & \ldots & \ddots &\vdots\\
e^{j\phi_{M,1}} & \ldots & e^{j\phi_{M,M-1}} & 1
\end{matrix}\right),
\een
which consists of $M^2-M$ phase shifters. As an illustrative example, Figure \ref{fig.psn} shows a PSN with $M=2$, wherein the switches are introduced so that the signals can bypass the PSN when the switches are off.

\begin{figure}[h!]
\centering
\includegraphics[width=3.5in]{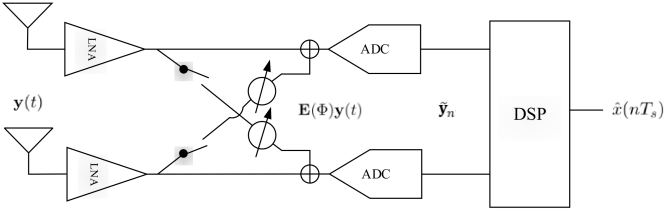}
\caption{Proposed phase shifter network at receiver (a 2x2 example)}
\label{fig.psn}
\end{figure}

When the switches are on, the output of the PSN is $\Ebf(\Phi) \ybf(t)$ and the output of the ADCs is  \cite{7307134}\cite{4407763}
\bena \label{eq.ybfn}
\tilde{\ybf}_n &\triangleq &{\cal Q}[\Ebf(\Phi)\ybf(nT_s)] \nonumber\\
 &=& (1-\rho)\Ebf(\Phi) \ybf(nT_s)+\qbf_n, \; n=1,2,\ldots,
\eena
where $\rho$ approximates to zero when the ADC bit resolution is larger than 4bit. Hence, we assume $\rho = 0$ in the rest of this paper. Note the difference between $\tilde{\ybf}_n$ in (\ref{eq.ybfn}) and $\ybf_n$ in (\ref{eqyn}): the PSN is engaged in the former case, it is bypassed in the latter.

To approximate the ideal prewhitener $\Rbf_y^{-1/2}$ using the PSN, it is natural to consider the following formulation [cf. (\ref{equ.RnegHalf})]
\begin{align}
\nonumber &\mathop{\min}_{\Phi, \Qbf}\ {||\Ebf(\Phi)-\Qbf\Sigmabf^{-1/2}\Ubf^{H}||^{2}_{F}} \\
&\text{subject to} \quad \phi_{mn} \in \Omega,\;  1\le m, n \le M
\label{costFuncOld}\\ \nonumber
& \hspace{1.7cm} \Qbf\Qbf^H = \Qbf^H\Qbf = \Ibf,
\end{align}
where the set $\Omega$ relies on the PS bit resolution:
\ben \label{eq.Omega}
\Omega = \left\{
\ba{ll} \left\{\frac{2k\pi}{2^{b}}, k = 0,\ldots, 2^b-1 \right\}, & \text{PS has $b$ bit resolution} \\
\left\{x \in {\mathbb R} | 0\le x < 2\pi\right\}, & \text{PS has $\infty$ resolution}.
\ea \right.
\een
We can solve (\ref{costFuncOld}) using an alternating method which is detailed in Appendix A. This formulation is intuitive and is introduced as a performance benchmark, although its performance is unsatisfactory as will be shown in the simulation.

\section{The HIMAP Scheme} \label{SEC3}

In this section, we first present an overview of the HIMAP scheme as a five-step procedure, before explaining in detail the two crucial steps: the PSN optimization algorithm and a constant false alarm rate (CFAR) detector.

\subsection{The Five Steps of HIMAP} \label{secFiveStep}
{\em Step 1}, turn off the PSN switches as illustrated in Figure \ref{fig.psn} so that the PSN is bypassed. Use $L_1$ samples of the ADC output to estimate the covariance matrix
\ben \hat{\Rbf}_y = \frac{1}{L_1} \sum_{n=1}^{L_1} {\ybf}_n \ybf_n^H, \label{eq.Ry1}
\een
where  ${\ybf}_n$ is the same as defined in (\ref{eqyn}).

In practice the quantization ${\cal Q}[\cdot]$ may render $\hat{\Rbf}_y$ rank-deficient, for which one remedy is to regularize $\hat{\Rbf}_y$ as
\ben
\hat{\Rbf}_y \leftarrow \hat{\Rbf}_y + \frac{\Delta^2}{12}\Ibf,
\label{equ.reg}
\een
where $\Delta$ is the quantization interval  of the ADCs and $\frac{\Delta^2}{12}$ is the variance of the quantization noise.


{\em Step 2}, optimize the prewhitening PSN matrix $\Ebf(\Phi)$ based on $\hat{\Rbf}_y$. It is the key step of the HIMAP and will be detailed in Section \ref{sec3.1}.

{\em Step 3}, set the phases of the PSN as obtained in {\em Step 2}, and turn on the switches to engage the PSN for analog prewhitening.\footnote{Here we omit the engineering detail that the input signals should be gain-controlled to match the dynamic range of the ADCs.}

{\em Step 4}, detect and synchronize the preamble of the frame now that the interferences are suppressed by the prewhitener. Here we adopt a CFAR detection method, which is to be explained in Section \ref{sec3.2}.

{\em Step 5}, compute the digital MMSE beamforming weight. After the synchronization, we can utilize the preamble sequence $\{ x_n, n=1,\ldots, L_2\}$ to estimate the effective array response
\ben
\tilde{\hbf}_{\rm eff} = \frac{1}{\sum_{n=1}^{L_2} |x_n|^2}\sum_{n=1}^{L_2} \tilde{\ybf}_n x^*_n \label{eq.heff}
\een
and the covariance matrix
\ben \tilde{\Rbf}_y = \frac{1}{L_2} \sum_{n=1}^{L_2} \tilde{\ybf}_n \tilde{\ybf}_n^H, \label{eq.Ry2}
\een
where $\tilde{\ybf}_n$ is the ADC output as given in (\ref{eq.ybfn}). Here the preamble sequence has already been synchronized.


The digital MMSE beamforming is [cf. (\ref{eqwmmse})]
\ben \label{wmmse}
\wbf_{\rm mmse} = \sigma_x^2 \tilde{\Rbf}_y^{-1} \tilde{\hbf}_{\rm eff},
\een
where $\sigma_x^2 = \frac{1}{L_2} \sum_{n=1}^{L_2} |x_n|^2 $.
Applying $\wbf_{\rm mmse}$ to the payload after the preamble, we can obtain
\ben
\tilde{x}_n = \wbf_{\rm mmse}^H \tilde{\ybf}_n, n = L_2+1, L_2+2,\dots,
\een
where the interferences have been suppressed.

\subsection{Optimization of PSN $\Ebf(\Phi)$} \label{sec3.1}
The entries of $\Ebf(\Phi)$ are constrained to unit modulus. As an ideal prewhitener should be proportional to $\hat{\Rbf}_y^{-1/2}$ to render the covariance matrix of the ADC inputs be a (scaled) identity matrix, the PSN-based prewhitening is usually non-ideal. But we attempt to optimize the PSN so that its output is spatially as white as possible.

When the switches are on, the covariance matrix of received signals is assumed as $\Ebf(\Phi)\hat{\Rbf}_y\Ebf(\Phi)$, and $\hat{\Rbf}_y$ is the covariance matrix estimated when the switches are off according to (\ref{eq.Ry1}) and (\ref{equ.reg}) in {\em step 1}. We propose to study the following optimization problem:
\begin{align}
\nonumber &\mathop{\max}_{\Phi}\ \frac{|\Ebf(\Phi)\hat{\Rbf}_y\Ebf(\Phi)^H|^{\frac{1}{M}}}{\tr(\Ebf(\Phi)\hat{\Rbf}_y\Ebf(\Phi)^H)/M} \\
&\text{subject to} \quad \phi_{ij} \in \Omega,\;  1\le i\ne j \le M \label{costFuncOri}\\
& \hspace{1.7cm} \phi_{ii} = 0, \; 1\le i \le M, \nonumber
\end{align}
where $|\cdot|$ represents the matrix determinant. The objective function, denoted as $\alpha(\Phi)$, is the ratio of the geometric mean to the algorithmic mean of the eigen-values of the matrix $\Ebf(\Phi) \Rbf_y \Ebf(\Phi)^H$. The higher the objective function, the more even the eigen-values are, and the less correlated the vector signal at the output of the PSN. Indeed, $\alpha(\Phi) \le 1$ with equality holds if and only if the eigen-values of $\Ebf(\Phi) \Rbf_y \Ebf(\Phi)^H$ are equal, i.e, it is a scaled identity matrix. That is, $\alpha(\Phi)$ is maximized when the exact prewhitening is achieved.

We can safely ignore the constraints $\phi_{ii} = 0, i=1,\ldots, M$ in (\ref{costFuncOri}), because replacing $\Ebf(\Phi)$ by ${\rm diag}(e^{-j\phi_{11}},\ldots, e^{-j\phi_{MM}}) \Ebf(\Phi)$ does not affect the objective function in (\ref{costFuncOri}).
We also have $|\Ebf(\Phi)\hat{\Rbf}_y\Ebf(\Phi)^H| = |\Ebf(\Phi)\Ebf(\Phi)^H||\hat{\Rbf}_y|$ since $\Ebf(\Phi)$ is square. Based on the above two observations, we can simplify (\ref{costFuncOri}) to be
\ben \label{eq.mainfunc}
\mathop{\max}_{\Phi \in \Omega^{M\times M}} \frac{|\Ebf(\Phi)\Ebf(\Phi)^H|^{\frac{1}{M}}}{\tr(\Ebf(\Phi)\hat{\Rbf}_y\Ebf(\Phi)^H)}.
\een

We present an iterative algorithm to solve the non-convex problem (\ref{eq.mainfunc}). Initialize the PSN with random phases $\Phi_0 \in \Omega^{M\times M}$. In the $i$th iteration, we sort $\Ebf(\Phi_{i})$ to be
\begin{equation}
\Ebf(\Phi_{i})=
\left[
\begin{matrix}
\xbf^H(\g{\phi}_l) \\
\bar{\Ebf}
\label{partitionequation}
\end{matrix}
\right],
\end{equation}
where the initial value of $\xbf^H(\g{\phi}_l)$ is $l$th row vector of $\Ebf(\Phi_{i-1})$ with $l=((i-1)\bmod M)+1$,
and the rest rows of $\Ebf(\Phi_{i-1})$ form $\bar{\Ebf}\in {\mathbb C}^{(M-1)\times M}$.

Given the partitioning, we have
\begin{equation}
\Ebf(\Phi_{i}) \Ebf(\Phi_{i})^H =
\left[
\begin{matrix} M & \xbf^H(\g{\phi}_l) \bar{\Ebf}^H \\
\bar{\Ebf} \xbf(\g{\phi}_l) & \bar{\Ebf} \bar{\Ebf}^H
\end{matrix}
\right].
\end{equation}
It follows that
\bena \label{eq.numer} &&|\Ebf(\Phi_{i}) \Ebf(\Phi_{i})^H| \\
&=&|\bar{\Ebf}\bar{\Ebf}^H|(M-
\xbf^H(\g{\phi}_l)\bar{\Ebf}^H(\bar{\Ebf}\bar{\Ebf}^H)^{-1}
\bar{\Ebf}\xbf(\g{\phi}_l)) \nonumber \\
&=&|\bar{\Ebf}\bar{\Ebf}^H| \xbf^H(\g{\phi}_l) (\Ibf-
\bar{\Ebf}^H(\bar{\Ebf}\bar{\Ebf}^H)^{-1}
\bar{\Ebf}) \xbf(\g{\phi}_l), \nonumber
\eena
while the denominator of (\ref{eq.mainfunc}) is
\ben \label{eq.denom}
\xbf^H(\g{\phi}_l) \hat{\Rbf}_y \xbf(\g{\phi}_l) + \tr(\bar{\Ebf}\hat{\Rbf}_y\bar{\Ebf}^H).
\een

Fixing the $\bar{\Ebf}$ and inserting (\ref{eq.numer}) and (\ref{eq.denom}) into (\ref{eq.mainfunc}) yields
\begin{equation}
\mathop{\max}_{\g{\phi} \in {\mathbb R}^M}\frac{[\xbf^H(\g{\phi_l})\Abf\xbf(\g{\phi_l})]^{\frac{1}{M}}}{\xbf^H(\g{\phi_l})\textbf{B}\xbf(\g{\phi_l})}
\label{costFuncSec}
\end{equation}
where
\begin{align} \label{eq.A}
&
\Abf \triangleq \Ibf-
\bar{\Ebf}^H(\bar{\Ebf}\bar{\Ebf}^H)^{-1}
\bar{\Ebf}, \\\label{eq.B}
&\textbf{B} \triangleq \hat{\Rbf}_y + \frac{\tr(\bar{\Ebf}\hat{\Rbf}_y\bar{\Ebf}^H)}{M}\textbf{I}.
\end{align}

Splitting $\xbf(\g{\phi}_l)$ as
\begin{equation}
\xbf(\g{\phi}_l) =
\left[
\begin{matrix}
0 \\
\vdots \\
0 \\
e^{j\phi_{l,n}} \\
0 \\
\vdots \\
0
\end{matrix}
\right]
+
\left[
\begin{matrix}
e^{j\phi_{l,1}} \\
\vdots \\
e^{j\phi_{l,n-1}} \\
0 \\
e^{j\phi_{l,n+1}} \\
\vdots \\
e^{j\phi_{l,M}}
\end{matrix}
\right]
= e^{j\phi_{l,n}}\ebf_n+\bar{\xbf}_{l,n},
\end{equation}
we can rewrite the cost function (\ref{costFuncSec}) as
\begin{equation}
\mathop{\max}_{\phi_n}\frac{[\ebf_n^H\Abf\ebf_n+2\Re(e^{-j\phi_{l,n}}\ebf_n^H\Abf\bar{\xbf}_{l,n})+\bar{\xbf}_{l,n}^H\Abf\bar{\xbf}_{l,n}]^{\frac{1}{M}}}
{\ebf_n^H\textbf{B}\ebf_n+2\Re(e^{-j\phi_{l,n}}\ebf_n^H\textbf{B}\bar{\xbf}_{l,n})+\bar{\xbf}_{l,n}^H\textbf{B}\bar{\xbf}_{l,n}},
\label{costFuncThi}
\end{equation}
which can be deduced as
\begin{equation}
\max_{\phi_{l,n}}g(\phi_{l,n})\triangleq\frac{[\alpha+r_1\text{cos}(\phi_{l,n}-\varphi_1
)]^{\frac{1}{M}}}{\beta+r_2\text{cos}(\phi_{l,n}-\varphi_2)}
\label{costFuncfor}
\end{equation}
where
\begin{equation}
\begin{split}
\alpha = a_{nn} + \bar{\xbf}_{l,n}^H\Abf\bar{\xbf}_{l,n}, &\quad \beta = b_{nn} + \bar{\xbf}_{l,n}^H\textbf{B}\bar{\xbf}_{l,n} \\
r_1 = 2|\ebf_n^H\Abf\bar{\xbf}_{l,n}|, &\quad r_2 = 2|\ebf_n^H\textbf{B}\bar{\xbf}_{l,n}| \\
\varphi_1 = \angle{\ebf_n^H\Abf\bar{\xbf}_{l,n}}, &\quad
\varphi_2 = \angle{\ebf_n^H\textbf{B}\bar{\xbf}_{l,n}}.
\end{split}
\end{equation}
Here $\angle(\cdot)$ stands for the phase of a complex number.

We have a closed-form solution to (\ref{costFuncfor}) for both finite and $\infty$ resolution of the phase shifters as is detailed in Appendix B.

Iterating through $n$ in the $i$th iteration and solving (\ref{costFuncfor}), we can improve the objective function of (\ref{costFuncOri}) monotonously until convergence, which will yield a prewhitening PSN $\Ebf(\Phi)$ as a (suboptimal) solution to (\ref{costFuncOri}). The optimality is not guaranteed owing to the non-convexity of the problem.

The whole procedure for solving (\ref{costFuncOri}), i.e., {\em Step} 2 of the HIMAP scheme is summarized in Algorithm 1.
\begin{algorithm}[ht]
\caption{The proposed algorithm for solving (\ref{costFuncOri})}
\begin{algorithmic}[1]
\REQUIRE The estimated covariance matrix $\hat{\Rbf}_y\in\mathbb{C}^{M \times M}$
\ENSURE The prewhitening PSN phases $\Phi \in \Omega^{M\times M}$
\STATE Initialize $\Ebf(\Phi)$ using some random phases.
\WHILE{the cost function in (\ref{eq.mainfunc}) still improves}
\FOR{$l = 1 : M$}
\STATE Take out the $l$th row of $\Ebf(\Phi)$ denoted as $\xbf^H(\phi)$ and  fix the other rows denoted as $\bar{\Ebf}^H$.
\STATE Compute (\ref{eq.A}) and (\ref{eq.B}).
\FOR{$n = 1 : M$}
\STATE Solve  $\phi_{l,n}$ in (\ref{costFuncfor}) according to the Appendix as the phase of the $(l,n)$th entry of $\Ebf$.
\ENDFOR
\ENDFOR
\ENDWHILE
\FOR{$l = 1 : M$}
\STATE $\Phi(l,:) = \Phi(l,:) -\phi_{l,l},$\footnotemark
\ENDFOR
\end{algorithmic}
\end{algorithm}

The computational complexity of Algorithm 1 is dominated by (\ref{eq.A}) and (\ref{eq.B}), both have complexity $O(M^3)$. But in each iteration we can update  $\tr(\bar{\Ebf}\hat{\Rbf}_y\bar{\Ebf}^H)$ as
\begin{equation} \tr(\bar{\Ebf}\hat{\Rbf}_y\bar{\Ebf}^H)-\xbf(\phibf_i)\hat{\Rbf}_y\xbf(\phibf_i)^H+\xbf(\phibf_j)\hat{\Rbf}_y\xbf(\phibf_j)^H,
\end{equation}
where $\xbf(\phibf_i)$ is the $i$th row of $\bar{\Ebf}$; hence the computational complexity of (\ref{eq.B}) is only $O(M^2)$. We can also simplify the update of (\ref{eq.A}) to only $O(M^2)$ flops by exploiting its low-rank property.

\footnotetext{Here we normalize the diagonal elements of $\Ebf(\Phi)$ to be $1$ according to (\ref{EphiStruc}).}
\subsection{The CFAR Detection}\label{sec3.2}
Now that the strong interferences have been significantly mitigated by the PSN prewhitener, the preamble can be synchronized in the digital domain. We use the detection and synchronization method from \cite{7880688}.

After the prewhitening PSN, the ADC outputs are $\tilde{\ybf}_n,n=1,2,\dots,N$. Detecting the presence of the synchronization sequence $\{x_n, n=1,\ldots, L_2\}$ is a hypothesis testing procedure. Here we choose \cite[\it Lemma 1]{7880688}
\ben \theta_p \triangleq \frac{\rbf(p)^H\Rbf(p)^{-1}\rbf(p)}{\sum_{l=1}^{L_2} |x_l|^2}\een
 as the testing metric with $p=0,1,\dots,N-L_2$ representing the time index.
\ben \label{equ.correla}
\begin{split}
\rbf(p) &= \sum_{n=1}^{L_2} \tilde{\ybf}_{p+n}x^*_n \in\Cnum^{M \times 1}
\end{split}
\een
is the cross-correlation  of the received sequence and the preamble sequence, which can be efficiently computed via fast fourier transform (FFT). The inverse of the covariance
\ben \label{autocorrelation}
\begin{split}
\Rbf(p) &= \sum_{n=1}^{L_2} \tilde{\ybf}_{p+n}\tilde{\ybf}^H_{p+n}
\end{split}
\een
can be computed recursively from $\Rbf^{-1}(p-1)$ using matrix inversion lemma since
\ben
\Rbf(p)=\Rbf(p-1)-\tilde{\ybf}_p\tilde{\ybf}_p^H
+\tilde{\ybf}_{p+L_2}\tilde{\ybf}_{p+L_2}^H.
\een

In absence of the preamble signal, it can be proven that $\theta$ is of $Beta(M,L_2-M)$ distribution \cite[\it Lemma 1]{7880688}
\ben
f_{\theta}(x)=\frac{(L_2-1)!}{(L_2-M-1)!(M-1)!}x^{M-1}(1-x)^{L_2-M-1},x\ge0.
\een
Indeed, the distribution does not depend on the variance of the interference-plus-noise. Hence, we can construct a CFAR detector by comparing $\theta_p$ with a threshold $\Gamma$, which is set according to a target false alarm rate (FAR) as
\ben
Pr(\theta\ge\Gamma)=\int_{\Gamma}^{\infty}f_{\theta}(x)dx = {\sf FAR}.
\een
If $\theta_{\bar{p}}\ge\Gamma$, then we can deem that the signal is in presence starting at time $\bar{p}+1$. We can further search $p$ around $\bar{p}$ to find a local maximum point as
\ben
p_{\rm sync} =  \mbox{arg}\max_{p\in [\bar{p}, \bar{p}+Q]} \theta_p,
\een
where $Q \ge 0$ is a chosen parameter.

Once the synchronization is achieved at time shift $p_{\rm sync}$, we can derive (\ref{wmmse}) from the results of (\ref{equ.correla}) and (\ref{autocorrelation}) as
\ben \label{eqwmmse2}
\hat{\wbf}_{\rm mmse}=\frac{1}{\frac{1}{L_2}\sum_{n=1}^{L_2}|x_n|^2}\Rbf^{-1}(p_{\rm sync})\rbf(p_{\rm sync}),
\een
which effectively completes both {\em Step 4} and {\em Step 5} of the HIMAP scheme.

To conclude this section, we note that the real PPSINR is quite different from $\rho_{\rm mmse}$ as given in (\ref{sinrOUT}) and should be computed as follows
\ben \label{sinrOUT2}
\rho_{\rm real} = \frac{|\hat{\wbf}^H_{\rm mmse} \Ebf(\Phi) \hbf |^2 \sigma_x^2} { \hat{\wbf}^H_{\rm mmse} \hat{\Qbf} \hat{\wbf}_{\rm mmse}}
\een
where $\hat{\wbf}_{\rm mmse}$ is as given in (\ref{eqwmmse2}) and
\ben
\hat{\Qbf} = {\mathbb E}[(\tilde{\ybf}_n -  \Ebf(\Phi) \hbf x(nT_s))(\tilde{\ybf}_n -  \Ebf(\Phi) \hbf x(nT_s))^H].
\een
In the simulations we use (\ref{sinrOUT2}) as the PPSINR metric.

\section{Numerical Examples} \label{SEC4}
In this section, we present simulation results to verify the effectiveness of the proposed HIMAP scheme. All the simulations are based on a receiver that has an $M$-element uniform linear array (ULA) with inter-element distance $d=\frac{\lambda}{2}$. While a signal with ${\sf SNR} = 25$dB impinges from the direction of arrival (DOA) $\theta_s = 0^\circ$, $K$ interferences impinge from $\theta_1,\dots, \theta_K$. For all the simulation except for the last one, assume line-of-sight (LOS) channel where the array response of the signal is
\begin{equation}
\hbf = \abf(\theta_s) \triangleq [1,e^{-j\pi \sin\theta_s},\dots,e^{-j(M-1)\sin\theta_s}]^T;
\label{equ.array}
\end{equation}
the last example simulates a Rayleigh fading channel. Throughout the simulations, $L_1=100$ samples are used for covariance estimation [cf. (\ref{eq.Ry1})] before prewhitening and another $L_2 = 100$ samples for estimating the effective channel vector and the covariance after prewhitening [cf. (\ref{eq.heff}) and (\ref{eq.Ry2})]. The ADC resolution is 12-bit unless stated otherwise.

The first simulation compares the performance of the two PSN prewhiteners---one is based on (\ref{costFuncOld}) and the other on (\ref{costFuncOri})---under two simulation settings: $M = 4$, $K = 1$ and $M = 4$, $K = 2$. One interference is from the angle  $\theta_1 = 60^\circ$; for $K=2$, the other interference is from $\theta_2=30^\circ$. As shown in Figure 2, the $\Ebf(\Phi)$ obtained by optimizing (\ref{costFuncOld}) yield output SINR (i.e.,pre-ADC SINR) significantly lower than that of the $\Ebf(\Phi)$ obtained according to (\ref{costFuncOri}). Therefore, we focus on the HIMAP scheme based on (\ref{costFuncOri}) in the remainder of this section.

\begin{figure}[htb]
\centering
\includegraphics[width=3.2in]{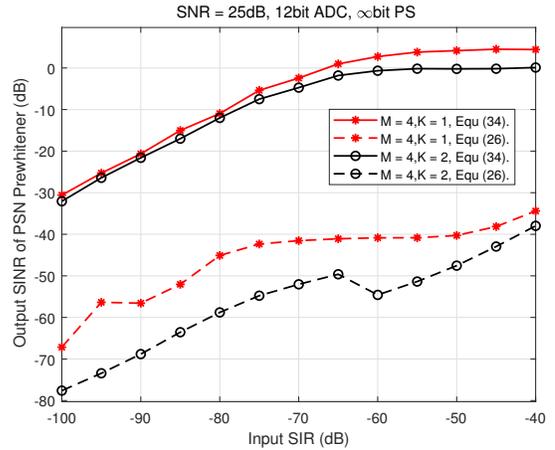}
\caption{PSN's output SINR-vs-input SIR performance of $\Ebf(\Phi)$ generated by  optimizing (\ref{costFuncOld}) and (\ref{costFuncOri}) with $M=4,K=1$ and $M=4,K=2$. ADC resolution = 12bit.}
\label{fig.comparison}
\end{figure}

In the second example, an interferences 70dB stronger than the signal impinges from angle $\theta_1=30^\circ$ on the two-element ULA ($M=2$). Based on the estimated covariance $\hat{\Rbf}$ as given in (\ref{eq.Ry1}), we run Algorithm 1 in Section \ref{sec3.1}. Figure \ref{fig.coverge1} shows the convergence of the proposed iterative algorithm for the PSN with different bit resolutions, where the y-axis is the objective function in (\ref{costFuncOri}) and one iteration represents the update of one row of $\Ebf(\Phi)$. We see that for the PSN with infinity resolution, the proposed algorithm can reach the global optimum within 10 iterations.
But the convergence to a global optimum  is not guaranteed for a general $M$. Figure \ref{fig.coverge2} illustrates the case where $M=4$, and two interfering sources are from $\theta_1 = 30^\circ$ and $\theta_2 = 60^\circ$, respectively. It shows that Algorithm 1 with different initializations converges to different local optimums.

\begin{figure}[htb]
\centering
\includegraphics[width=3.2in]{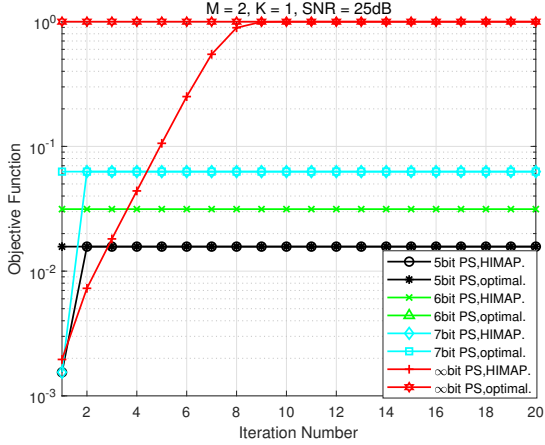}
\caption{Convergence  of objective function (\ref{costFuncOri}) when running Algorithm 1, $M=2$, $K=1$, ADC resolution = 12bit, input ${\sf SIR}=-70\text{dB}$.}
\label{fig.coverge1}
\end{figure}

\begin{figure}[htb]
\centering
\includegraphics[width=3.2in]{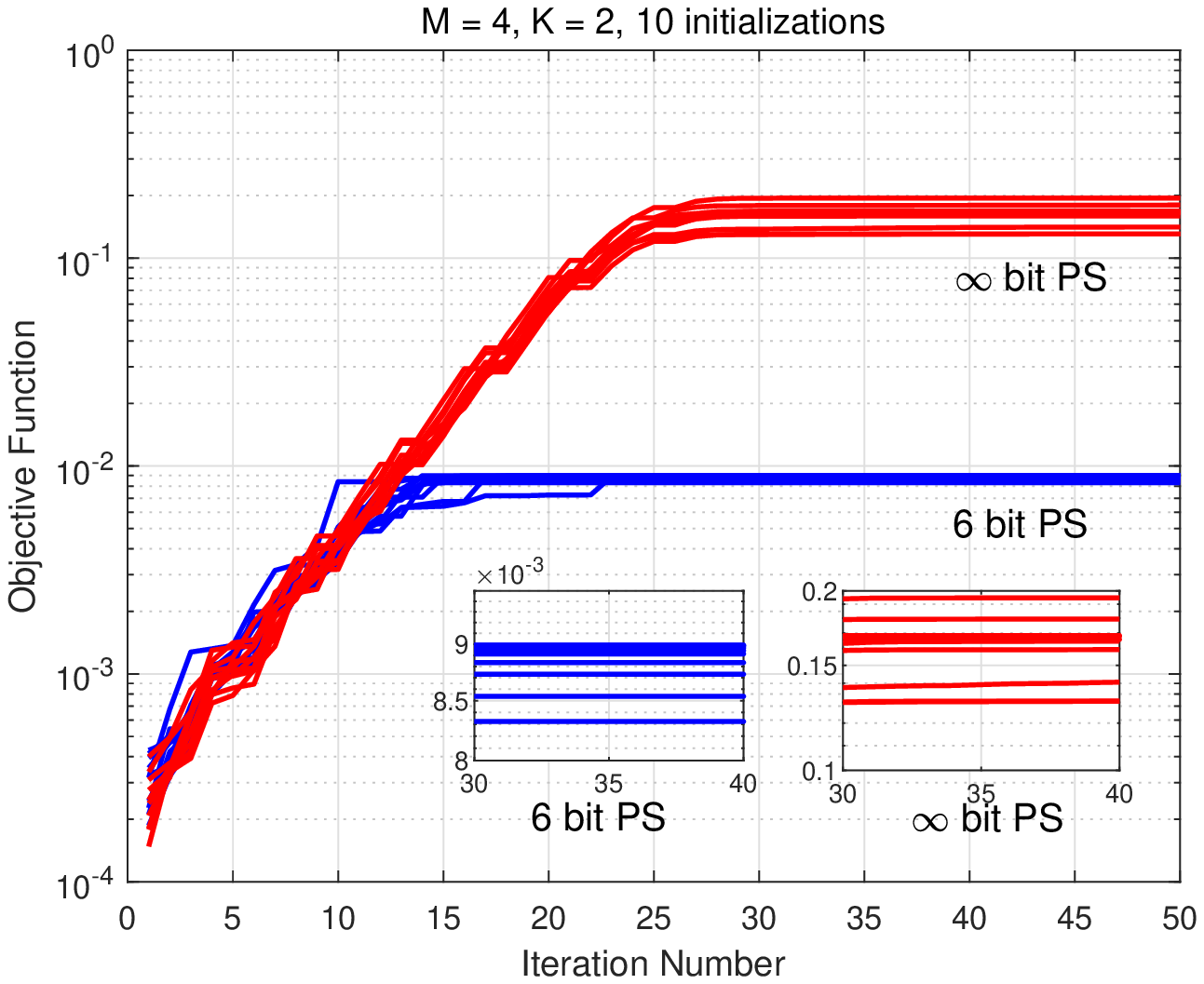}
\caption{Convergence  of objective function of (\ref{costFuncOri}) when running Algorithm 1, $M=4$, $K=2$,  ADC resolution = 12bit, input ${\sf SIR}=-70\text{dB}$.}
\label{fig.coverge2}
\end{figure}

The third example simulates the same scenario as Figure \ref{fig.coverge1} but with varying input SIR. Figure \ref{fig.SINRbfADC} shows the SINR of the input into the ADC [cf. (\ref{eq.pre-ADC})]
$$
\rho_{\rm prewhiten} = \frac{\|\Ebf(\Phi)\abf(\theta_s)\|^2 10^{2.5}}{\tr\{\Ebf(\Phi)[\abf(\theta_1) \abf^H(\theta_1) 10^{\frac{\gamma}{10}}+ \Ibf] \Ebf(\Phi)\}}
$$
where $\gamma = 25-{\sf SIR}_{\rm dB}$ is interference-to-noise ratio (INR) in dB. We optimize  the  PSN of $6$-bit, $7$-bit, and $\infty$-bit resolution, respectively. For comparison purposes, we also simulate the ideal prewhitener $\hat{\Rbf}_y^{-1/2}$. The performance of using $\infty$-bit PSN coincides with the ideal prewhitener in this particular case, which agrees with the fact that the red line {$-\hspace{-0.2em}+\hspace{-0.2em}-$} converges to 1 in Figure \ref{fig.coverge1}. The ideal prewhitening yields SINR 0dB, which agrees with the comment following Lemma \ref{lemma1} that $\rho_{\rm prewhite} \approx \frac{1}{M-1}$ if $\rho_{\rm mmse} \gg 1$. Here $M=2$, so $\rho_{\rm prewhite} \approx 1 = 0$dB. Using the off-the-shelf phase shifters of $6$-bit resolution can mitigate interferences by 25dB as shown in Figure \ref{fig.SINRbfADC}.

\begin{figure}[htb]
\centering
\includegraphics[width=3.2in]{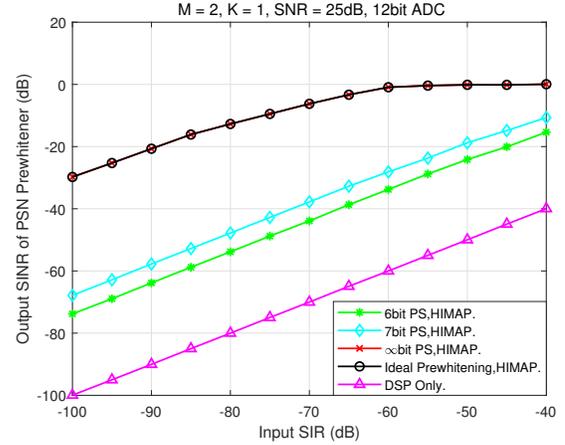}
\caption{PSN's output SINR-vs-input SIR performance of different methods.}
\label{fig.SINRbfADC}
\end{figure}


The fourth example simulates the case where the ULA has $M=4$ antennas and the $K=2$ interferences impinging from angles $\theta_1=30^\circ$ and $\theta_2=60^\circ$, respectively. Figure \ref{fig.M4K2DiffSINR} shows the PPSINR of the HIMAP with PSN of different bit resolutions under different SIR. Here PPSINR is computed based on (\ref{sinrOUT2}) rather than (\ref{sinrOUT}).
The gain of the HIMAP over the DSP-only method is prominent and is more so when using PS' of higher resolution. As a benchmark, we also include the case of $\infty$-resolution ADCs, with which the MMSE receiver yields PPSINR as given in (\ref{sinrOUT}).

\begin{figure}[tb]
\centering
\includegraphics[width=3.2in]{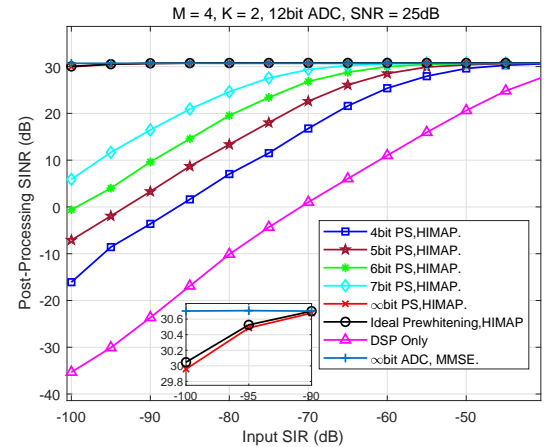}
\caption{PPSINR performance comparison as function of input SIR among the proposed method, ideal prewhitening and DSP only, and $M=4$,$K=2$.}
\label{fig.M4K2DiffSINR}
\end{figure}

%

\begin{figure}[htb]
\centering
\includegraphics[width=3.2in]{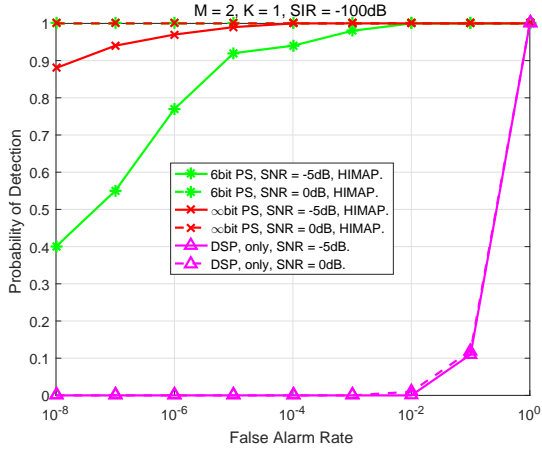}
\caption{Probability of detection-vs-false alarm rate performance at SNR 0dB and -5dB.}
\label{fig.synch}
\end{figure}
The fifth example simulates the performance of preamble detection with the input $\sf SIR=-100dB$, $M = 2, K = 1$. Figure \ref{fig.synch} shows the receiver operating characteristic (ROC) performance, i.e., the probability of detection (PD) versus false alarm rate (FAR), of the PSN-based HIMAP schemes (6-bit PSN, $\infty$-bit PSN), and the DSP-only method. Two SNR settings are used: -5dB (the solid lines) and 0dB (the dash lines). The HIMAP sees prominent improvement of the detection performance over the DSP-only method. The HIMAP scheme using either 6-bit or $\infty$-bit PSN can achieve 100\% detection at SNR=0dB as the dash lines overlaps at the $PD=1$ level, while the DSP-only method cannot because of the large quantization noise. Indeed, the SQNR is $-27$dB when the input SIR is $-100$dB according to (\ref{eq.SQNRdb}), which cannot be sufficiently compensated by the 20dB processing gain of the $L_2=100$ length preamble.

\begin{figure}[tb]
\centering
\includegraphics[width=3.2in]{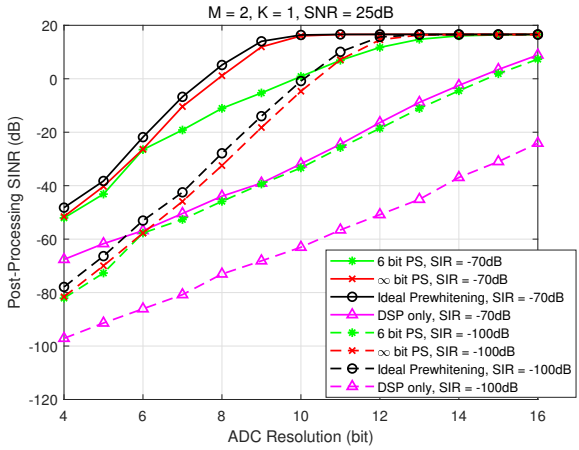}
\caption{PPSINR  of the schemes for ADCs of  different bit-resolutions.}
\label{fig.M2K1fu100fu70}
\end{figure}

The sixth example compares the PPSINR performance of the HIMAP schemes and the DSP-only method applied to the two-element ULA receiver, given ADCs of  different bit resolutions. Two SIR settings are simulated, $\sf SIR = - 70$dB and $\sf SIR = - 100$dB, which correspond to the solid lines and dash lines in Figure \ref{fig.M2K1fu100fu70}, respectively. It is shown that even with the PSN of only 6-bit resolution and the ADC of 11-bit ENOB, the HIMAP scheme can suppress 70dB interference. The simulation results show that the HIMAP using the 2x2 PSN can cut the requirement of ADC resolution by 4$\sim$5 bits, which amounts to huge reduction of power consumption. For instance, reducing the ADC bit resolution from 16 to 12 amount to $4^{16-12} = 256$x power reduction (4x per extra ENOB) \cite{7258468}. 
\begin{figure}[tb]
\centering
\includegraphics[width=3.2in]{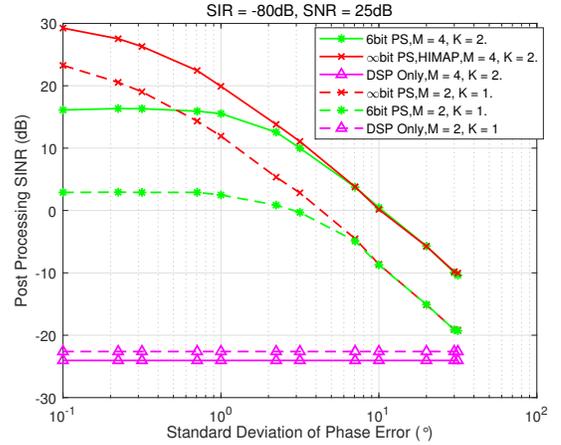}
\caption{Comparison of PPSINR  of the schemes under different ADC resolution, input ${\sf SIR}=-80\text{dB}$.}
\label{fig.M2K1PhsErr1}
\end{figure}

To evaluate impact of hardware non-perfectness, we simulate phase error of the PSN. Given the nominal phase $\phi$ of the phase shifters, we model the actual phase $\hat{\phi}$ as a Gaussian random variable $\hat{\phi} \sim N(\phi,\sigma^2)$. Figure \ref{fig.M2K1PhsErr1} shows the PPSINR performance versus the standard deviation  $\sigma$ in two scenarios: $M=2$, $K=1$ and $M=4$, $K=2$, which correspond to the dash lines and solid lines, respectively. We can see that the $\infty$-bit PSN-based HIMAP scheme is more sensitive to phase errors than the 6-bit PSN-based HIMAP scheme. The latter can tolerate phase error of standard deviation $\sigma = 1^\circ$. This result suggests that the PSN needs to be calibrated to prevent performance deterioration.

\begin{figure}[tb]
\centering
\includegraphics[width=3.2in]{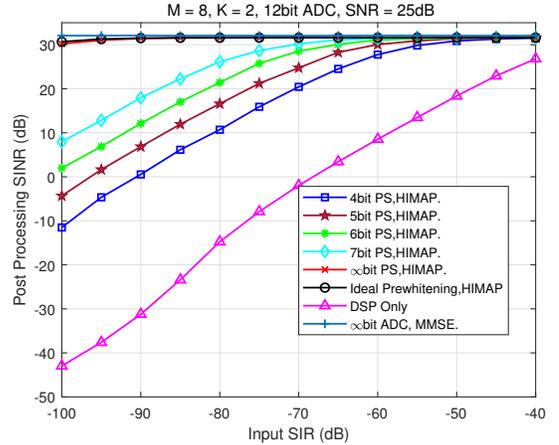}
\caption{PPSINR performance comparison as function of input SIR among the proposed method, ideal prewhitening and DSP only in Rayleigh fading channel, and $M=8$, $K=2$.}
\label{fig.rayleighFading}
\end{figure}

The last example simulates the case where both the signal and interferences are through a Rayleigh fading channel, owing to abundance of multipath scatterings. Figure \ref{fig.rayleighFading} shows the PPSINR performance of the HIMAP with PSN of different bit resolutions. Figure \ref{fig.rayleighFading} shares same settings with Figure \ref{fig.M4K2DiffSINR}, except that here the antenna number $M = 8$. The gain of the HIMAP scheme remains dramatic compared with the DSP only method, which verifies the universal feasibility of the HIMAP scheme in different channel environments.



\section{Conclusion} \label{SEC5}
In this paper, we propose a scheme named hybrid interference mitigation using analog prewhitening (HIMAP), which employs an $M\times M$ phase shifter network (PSN) inserted between the antenna ports and the analog-to-digital converters (ADC). Using only the covariance matrix estimate, the HIMAP scheme can optimizes the PSN to mitigate the interferences via spatial prewhitening, which help significantly reduce ADCs' quantization noise. Through the combination of the {\em analog} PSN  prewhitener and the {\em digital} MMSE beamformer, the HIMAP scheme can suppress strong interferences using off-the-shelf ADCs as verified by the simulations. Since the HIMAP assumes no information of the interferences, it works for both non-cooperative interferences and collocated interferences in full-duplex wireless. The simulation also shows how sensitive the scheme is to the phase errors of the PSN. The high-precision calibration of the PSN can be an interesting future research topic.

\section*{Appendix A: An Alternating Method to Solve (\ref{costFuncOld})}

First, we initialize $\Qbf$ as an arbitrary unitary matrix; then $\Ebf(\Phi)$ can be obtained as
\ben
[\Ebf(\Phi)]_{mn} = e^{j\angle{[\Cbf]_{mn}}},m,n = 1,2,\dots,M.
\een
where $\Cbf = \Qbf\Sigmabf^{-1/2}\Ubf^{H}$ and $[\cdot]_{mn}$ denotes the $(m,n)$th element of the matrix.

Second, with $\Ebf(\Phi)$ being fixed, to minimize (\ref{costFuncOld}) with respect to $\Qbf$ is an orthogonal procrustes problem (OPP), whose solution is \cite{schonemann1966a}
\ben
\Qbf = \tilde{\Vbf}\tilde{\Ubf}^H,
\een
where $\tilde{\Vbf}$ and $\tilde{\Ubf}$ are from the SVD $\Sigmabf^{-1/2}\Ubf^{H}\Ebf(\Phi)^H = \tilde{\Ubf}\tilde{\Sigma}\tilde{\Vbf}^H$. By iterating across the variables $\Ebf(\Phi)$ and $\Qbf$ until (\ref{costFuncOld}) convergence, we can obtain $\Ebf(\Phi)$ as a prewhitening matrix, albeit a suboptimal one.

\section*{Appendix B: A Closed-form Solution to (\ref{costFuncfor})} \label{der.closeform}
We rewrite (\ref{costFuncfor}) and replace $\phi_{l,n}$ by $\phi$ for notational simplicity in the below
\ben
\mathop{\text{max}}_{\phi}g(\phi)= \frac{f^{\frac{1}{M}}(\phi)}{h(\phi)},
\label{equ.costfuncAppend}
\een
where
\ben
\begin{split}
f(\phi)\triangleq\alpha+r_1\text{cos}(\phi-\varphi_1
), \\
h(\phi)\triangleq\beta+r_2\text{cos}(\phi-\varphi_2
).
\end{split}
\label{equ.fandh}
\een
Equating $\frac{dg}{d\phi}=0$ yields
\ben
f^\prime(\phi)h(\phi)-Mf(\phi)h^\prime(\phi)=0.
\label{equ.der}
\een
Denoting $z=\text{tan}(\frac{\phi}{2})$, we rewrite (\ref{equ.der}) as
\ben
z^4+a_3z^3+a_2z^2+a_1z+a_0=0
\label{equ.poly}
\een
with
\ben
\begin{split}
a_3=&\frac{2v_1\beta-2Mv_2\alpha+(2M-2)(v_1v_2-u_1u_2)}{u_1(\beta-v_2)-Mu_2(\alpha-v_1)},\\
a_2=&\frac{(4-2M)u_2v_1+(2-4M)u_1v_2}{u_1(\beta-v_2)-Mu_2(\alpha-v_1)},\\
a_1=&\frac{2v_1\beta-2Mv_2\alpha+(2M-2)(u_1u_2-v_1v_2)}{u_1(\beta-v_2)-Mu_2(\alpha-v_1)},\\
a_0=&\frac{Mu_2(\alpha+v_1)-u_1(\beta+v_2)}{u_1(\beta-v_2)-Mu_2(\alpha-v_1)},
\end{split}
\een
where
\ben
\begin{split}
u_1=r_1\text{sin}(\varphi_1),\quad &v_1=r_1\text{cos}(\varphi_1),\\
u_2=r_2\text{sin}(\varphi_2),\quad &v_2=r_2\text{cos}(\varphi_2).
\end{split}
\een
The quartic equation (\ref{equ.poly}) has four roots of closed-form
\ben \label{solu.poly}
\begin{split}
z_1,z_2=-\frac{1}{4}a_3+\frac{1}{2}R\pm\frac{1}{2}D, \\
z_3,z_4=-\frac{1}{4}a_3-\frac{1}{2}R\pm\frac{1}{2}E,
\end{split}
\een
where
\begin{small}
\ben
\begin{split}
R=&\sqrt{\frac{1}{4}a_3^2-a_2+t_1}, \\
D=&\left\{
\begin{aligned}
&\sqrt{\frac{3}{4}a_3^2-R^2-2a_2+\frac{1}{4}(4a_3a_2-8a_1-a_3^3)R^{-1}},&R\ne0 \\
&\sqrt{\frac{3}{4}a_3^2-2a_2+2\sqrt{t_1^2-4a_0}},&R=0
\end{aligned}
\right. \\
E=&\left\{
\begin{aligned}
\footnotesize
&\sqrt{\frac{3}{4}a_3^2-R^2-2a_2-\frac{1}{4}(4a_3a_2-8a_1-a_3^3)R^{-1}},&R\ne0 \\
&\sqrt{\frac{3}{4}a_3^2-2a_2-2\sqrt{t_1^2-4a_0}},&R=0
\end{aligned}
\right.
\end{split}
\een
\end{small}
and $t_1$ is a real-valued root of the cubic equation
\ben \label{equ.cube}
t^3-a_2t^2+(a_1a_3-4a_0)t+(4a_2a_0-a_1^2-a_3^2a_0)=0,
\een
which also has analytic solutions \cite{abramowitz1964handbook}.

Denoting ${\cal Z}$ as the set of real-valued $z_i$'s in (\ref{solu.poly}), we obtain the set of the potential solutions to (\ref{equ.costfuncAppend}) as $\Xi=\{\phi|\phi=2\text{arctan}(z_i),z_i\in{\cal Z}\}$. Thus, for phase shifter of $\infty$
resolution, the solution to (\ref{costFuncfor}) is
\ben
\phi_{\infty}=\mbox{arg}\mathop{\text{max}}_{\phi \in \Xi} g(\phi),
\label{equ.infty}
\een
which is easy as the cardinality $|\Xi| \le 4$.

For phase shifters of $b$-bit resolution, we show in the next that the solution to (\ref{costFuncfor})
\ben \label{equ.finitySet}
\phi_b=\mbox{arg}\mathop{\text{max}}_{\phi \in \Omega} g(\phi),
\een
can be found by only checking two points in $\Omega$ neighboring to $\phi_{\infty}$.

\begin{figure}[htb]
\centering
\includegraphics[width=3.5in]{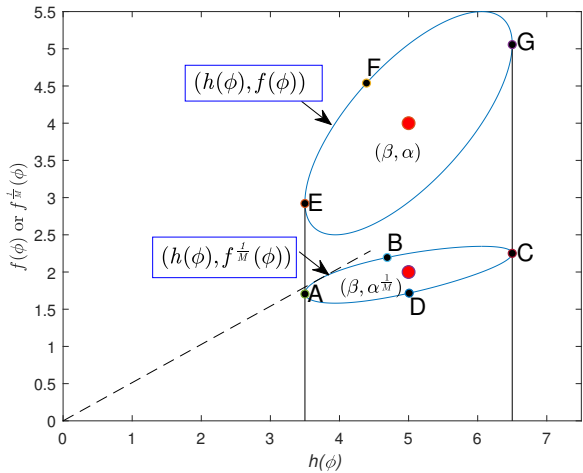}
\caption{The traces of ($h(\phi)$,$f(\phi)$) and ($h(\phi)$,$f^{\frac{1}{M}}(\phi)$) as $\phi$ varies from 0 to $2\pi$. Here we illustrate the case of $M=2$.}
\label{fig.trace}
\end{figure}

Figure \ref{fig.trace} illustrates the trace of $(h(\phi),f(\phi))$ and $(h(\phi),f^{\frac{1}{M}}(\phi))$ with
$h(\phi)$ and $f(\phi)$ being the x-axis and y-axis as $\phi$ varies from $0$ to $2\pi$, for which we see that to solve (\ref{equ.costfuncAppend}) amounts to finding a point $P$ on the shape of $ABCD$ so that the line connecting the origin and $P$ has the steepest slope.

The trace of $(h(\phi),f(\phi))$ forms an ellipsoid centered at $(\beta,\alpha)$ according to (\ref{equ.fandh}), and the arc $\wideparen{EFG}$ of the ellipsoid can be represented by $\eta_1$ as a function of $h(\phi)$
\ben
\eta_1(h(\phi)), \beta-r_2\le h(\phi)\le\beta+r_2,
\een
which is a concave function of $h(\phi)$, Thus,
(\ref{equ.costfuncAppend}) can be formulated as
\ben
\mathop{\text{max}}_{h\in[\beta-r_2,\beta+r_2]}\eta_2(h)\triangleq \frac{\eta_1^{\frac{1}{M}}(h)}{h},
\label{equ.hinfty}
\een
where $h(\phi)$ is replaced with $h$ for notational simplicity.
\begin{Lemma}\label{lemma2}
$\eta_2(h)$ has only one maximum at $ h_0\in (\beta-r_2, \beta+r_2)$, and it is monotonously increasing for $h\in[\beta-r_2, h_0]$ and is monotonously decreasing for $h\in[ h_0,\beta+r_2]$.
\end{Lemma}
\begin{proof}
 The first-order derivative  of $\eta_2(h)$ is
\ben
\frac{d\eta_2(h)}{dh}=\frac{\frac{1}{M}\eta_1^{\frac{1}{M}-1}(h)\eta_1^\prime(h)h-\eta_1^{\frac{1}{M}}(h)}{h^2}.
\label{equ.firDer}
\een
Denote the numerator of (\ref{equ.firDer}) as
\ben
\eta_3(h)\triangleq\frac{1}{M}\eta_1^{\frac{1}{M}-1}(h)\eta_1^\prime(h)h-\eta_1^{\frac{1}{M}}(h).
\een
Then the first-order derivative
\begin{align}
\frac{d\eta_3(h)}{dh} =& \left(\frac{1}{M^2}-\frac{1}{M}\right)\eta_1^{\frac{1}{M}-2}(h)(\eta_1^\prime(h))^2h \nonumber\\
&
+\frac{1}{M}\eta_1^{\frac{1}{M}-1}(h)\eta_1^{\prime \prime}(h)h \le 0, \label{equ.deta3}
\end{align}
since
\ben
\frac{1}{M^2}-\frac{1}{M} <0,\, \eta_1(h)\ge0,\, \eta_1^{''}(h)\le0, \mbox{ and } h\ge0.
\een
Thus, $\eta_3(h)$ is a monotonously decreasing function as $h$ varies from $\beta-r_2$ to $\beta+r_2$. Moreover, we have
\ben
\begin{split}
\lim_{h \to (\beta-r_2)^{+}}\eta_3(h)=+\infty, \\
\lim_{h \to (\beta+r_2)^{-}}\eta_3(h)=-\infty,
\end{split}
\label{equ.limit}
\een
since
\ben
\begin{split}
\lim_{h \to (\beta-r_2)^{+}}\eta_1^\prime(h)=+\infty, \\
\lim_{h \to (\beta+r_2)^{-}}\eta_1^\prime(h)=-\infty.
\end{split}
\een
Combining (\ref{equ.deta3}) and (\ref{equ.limit}), we can conclude that the equation $\eta_3(h)=0$ has only one root  $ h_0\in [\beta-r_2, \beta+r_2]$ and
\ben
\begin{split}
\eta_3(h)\ge0, \beta-r_2\le h\le h_0, \\
\eta_3(h)\le0,  h_0< h\le \beta+r_2,
\end{split}
\een
which indicates that $\eta_2(h)$ is monotonously increasing for $h\in[\beta-r_2, h_0]$ while it's monotonously decreasing for $h\in[ h_0,\beta+r_2]$. Thus, $ h_0$ is the only maximum point of $\eta_2(h)$.

\end{proof}
If phase shifter's bit resolutions $b<\infty$ and $\phi\in\Omega$, $h$ lies in the set
\ben
\Theta = \{h|h=\beta+r_2\text{cos}(\phi-\varphi_2),\phi\in\Omega\},
\een
which turns (\ref{equ.hinfty}) into
\ben
\mathop{\text{max}}_{h\in\Theta}\eta_2(h)=\frac{\eta_1^{\frac{1}{M}}(h)}{h}.
\label{equ.hfty}
\een

Denote
\ben
 h_0 \triangleq \beta+r_2\text{cos}(\phi_{\infty}-\varphi_2),
\een
where $\phi_{\infty}$ is the solution to (\ref{equ.infty}).
Due to the monotonicity of $\eta_2(h)$ proved in Lemma \ref{lemma2}, the solution of (\ref{equ.hfty}) must be in the set $\{h_l,h_r\}$
where
\ben
h_l=\mathop{\text{max}}_{h\le h_0,h\in\Theta}h, \quad h_r=\mathop{\text{min}}_{h\ge h_0,h\in\Theta}h,
\label{equ.landr}
\een
and $h$ is a monotonous function of $\phi$ as $h$ varies from $\beta-r_2$ to $\beta+r_2$ on the arc $\wideparen{ABC}$.

Thus, solving (\ref{equ.landr}) is equivalent to finding an integer $\tilde{k}$ that satisfies
\ben
\frac{2(\tilde{k}-1)\pi}{2^b}\le\phi_{\infty}\le\frac{2\tilde{k}\pi}{2^b}.
\een
Then the solution to (\ref{equ.finitySet}) is simply
\ben
\phi_b = \left\{
\ba{ll} \frac{2(\tilde{k}-1)\pi}{2^b}, & g(\frac{2(\tilde{k}-1)\pi}{2^b})>g(\frac{2\tilde{k}\pi}{2^b}). \\
\frac{2\tilde{k}\pi}{2^b}, & \text{otherwise}.
\ea \right.
\label{equ.finitySet2}
\een
where $g(\cdot)$ is as defined in (\ref{equ.costfuncAppend}).

\bibliographystyle{ieeetran}
\bibliography{bib}
\end{document}

%% file: com.tex
\newcommand{\ba}{\begin{array}}
\newcommand{\ea}{\end{array}}
\newcommand{\be}{\begin{displaymath}}
\newcommand{\ee}{\end{displaymath}}
\newcommand{\ben}{\begin{equation}}
\newcommand{\een}{\end{equation}}
\newcommand{\bena}{\begin{eqnarray}}
\newcommand{\eena}{\end{eqnarray}}
\newcommand{\beqa}{\begin{eqnarray*}}
\newcommand{\enqa}{\end{eqnarray*}}

\newcommand{\bc}{\begin{center}}
\newcommand{\ec}{\end{center}}
\newcommand{\bi}{\begin{itemize}}
\newcommand{\ei}{\end{itemize}}
\newcommand{\benu}{\begin{enumerate}}
\newcommand{\eenu}{\end{enumerate}}
\newcommand{\bdes}{\begin{description}}
\newcommand{\edes}{\end{description}}
\newcommand{\bt}{\begin{tabular}}
\newcommand{\et}{\end{tabular}}

\newcommand \phibf{\mbox{\boldmath$\phi$\unboldmath}}

\newcommand \Sigmabf{\hbox{$\bf \Sigma$}}

\newcommand \abf{{\bf a}}

\newcommand \ebf{{\bf e}}

\newcommand \gbf{{\bf g}}
\newcommand \hbf{{\bf h}}

\newcommand \qbf{{\bf q}}
\newcommand \rbf{{\bf r}}

\newcommand \vbf{{\bf v}}
\newcommand \wbf{{\bf w}}
\newcommand \xbf{{\bf x}}
\newcommand \ybf{{\bf y}}
\newcommand \zbf{{\bf z}}

\newcommand \Abf{{\bf A}}

\newcommand \Cbf{{\bf C}}

\newcommand \Ebf{{\bf E}}

\newcommand \Ibf{{\bf I}}

\newcommand \Qbf{{\bf Q}}
\newcommand \Rbf{{\bf R}}

\newcommand \Ubf{{\bf U}}
\newcommand \Vbf{{\bf V}}

\newcommand{\Cnum}{{\mathbb C}}

\newcommand{\Enum}{{\mathbb E}}

\newcommand{\circlambda}{\mbox{$\Lambda$
             \kern-.85em\raise1.5ex
             \hbox{$\scriptstyle{\circ}$}}\,}

\newcommand{\tr}{\mathop{\rm tr}}

\def\Re{\mathop{\rm Re}}

\newtheorem{Theorem}{Theorem}[section]

\newtheorem{Lemma}[Theorem]{Lemma}

